\documentclass[11pt]{article}
\usepackage[margin=1in]{geometry}

\usepackage{amsmath}
\usepackage{amssymb}
\usepackage{amsthm}
\usepackage{hyperref}
\usepackage{graphicx} 
\usepackage[algosection,ruled,lined,linesnumbered,longend]{algorithm2e}
\usepackage{algorithmic}
\usepackage{latexsym}
\usepackage{graphicx,wrapfig,xcolor}
\usepackage{float}
\usepackage{caption}
\usepackage{subcaption}
\usepackage{comment}
\usepackage{xcolor}
\usepackage{algorithm2e}
\usepackage{wrapfig}
\usepackage{xspace}

\usepackage{todonotes}

\newcommand{\beq}{\begin{equation}}
\newcommand{\bet}{\begin{table}}
\newcommand{\eeq}{\end{equation}}

\newcommand{\E}{\mathbb{E}}

 \newtheorem{theorem}{Theorem}[section]
 \newtheorem{lemma}[theorem]{Lemma}

\newcommand{\defn}[1]{{\textit{\textbf{\boldmath #1}}}}

\usepackage{hyperref}
\usepackage{cleveref}
\usepackage{MnSymbol}

\newcommand{\stime}{time\xspace}
\newcommand{\stimes}{times\xspace}
\newcommand{\volume}{max-coordinate-volume\xspace}

\title{Multidimensional Resource Scheduling with Small Demands%
\thanks{Part of this work was done while the authors were visiting the
Simons Institute for the Theory of Computing at UC Berkeley during the
Fall 2025 research program \emph{Algorithmic Foundations for Emerging
Computing Technologies}.}}
\date{}

\author{Yossi Azar\thanks{Tel Aviv University, Israel. \texttt{azar@tau.ac.il}.
Yossi Azar was partially supported by ISF
Grant 1728/26.} 
\and Rathish Das\thanks{University of Houston, USA. \texttt{rathish@uh.edu}.} 
\and Hao Sun\thanks{University of Houston, USA. \texttt{hsun33@central.uh.edu}.}}

\begin{document}
\maketitle
\thispagestyle{empty}
\begin{abstract}

We study multidimensional resource scheduling. Each job \(i\) has a \(d\)-dimensional resource-demand vector \(v_i\) and a processing time \(s_i\). The scheduler assigns a start time to each job, subject to the constraint that, at every time, the total demand of the jobs being processed does not exceed \(1\) in any resource dimension. The objective is to minimize the makespan.

We focus on the regime in which every individual resource demand is small. We ask whether the favorable \emph{small-vector phenomenon} known for multidimensional vector packing, which corresponds to the special case of unit processing times, extends to jobs with heterogeneous processing times. The key difficulty is that arbitrary processing times create temporal interactions across multiple duration scales: a long job may overlap many shorter jobs, while feasibility must be maintained throughout every job's execution interval.

We prove that the small-vector phenomenon persists in this temporal setting. For any 
\(0<\epsilon<1/4\), after normalizing the maximum processing time to \(1\), if every coordinate of every demand vector is at most \(O(\epsilon^2/\log(d/\epsilon))\), we give a randomized offline algorithm that produces a schedule with expected makespan at most \(
(1+6\epsilon)\mathrm{OPT}+3 \).
Thus, sufficiently small resource demands admit asymptotically near-optimal schedules in arbitrary dimension, despite heterogeneous processing times.

We also obtain constant competitive ratios in the online setting. Let \(T\) denote the ratio between the maximum and minimum processing times. If every coordinate is at most \(O(1/(\log d\log T))\), we give a randomized \(O(1)\)-competitive algorithm, with a competitive ratio independent of both \(d\) and \(T\). We further derandomize our approach, obtaining a deterministic \(O(1)\)-competitive algorithm under a comparable smallness assumption.
\end{abstract}

\newpage
\setcounter{page}{1} 
\section{Introduction}
Modern computing systems execute many jobs concurrently while sharing several
limited resources.  A job may simultaneously consume CPU capacity, memory,
memory bandwidth, I/O bandwidth, network capacity, or accelerator resources,
and these requirements persist throughout the time that the job executes.
Representing such a job by a single scalar size can therefore hide the actual
source of contention: two jobs with comparable aggregate demand may interact
very differently depending on which resources they use.  This motivates
\emph{multidimensional resource scheduling}, where the resource requirement of
a job is represented by a vector rather than a scalar.

In this paper, we study the following fundamental form of this problem.
There are \(n\) independent jobs.  Job \(i\) has a processing time \(s_i>0\)
and a resource vector
\[
    v_i=(v_{i,1},\ldots,v_{i,d})\in\mathbb{R}_{+}^{d}.
\]
There is unit capacity in each of the \(d\) resources.  Once job \(i\) starts,
it occupies \(v_{i,k}\) units of resource \(k\) throughout its entire
processing interval.  Thus, if \(f(i)\) is its start time, a schedule is
feasible if, for every time \(t\) and every resource \(k\in[d]\),
\[
    \sum_{i:\,f(i)\le t<f(i)+s_i} v_{i,k}\le 1.
\]
The objective is to minimize the makespan
\[
    C_{\max}=\max_i\{f(i)+s_i\}.
\]
Geometrically, the time axis forms a strip, while every job occupies an
interval of length \(s_i\) and consumes a \(d\)-dimensional vector of
capacities throughout that interval.  We therefore also view the problem as a
multidimensional strip-packing problem. When the processing times of jobs are the same, our problem is equivalent to the well-studied vector-bin-packing problem~\cite{ChekuriKhanna2004}.

\paragraph{\textbf{Why is the problem important?}}
Multidimensional resource requirements arise naturally whenever jobs execute
concurrently on a shared system.  Classical work on multidimensional
scheduling was already motivated by tasks requiring several resources such as
CPU, disks, and network capacity simultaneously
\cite{ChekuriKhanna2004}.  This issue has only become more pronounced as
modern workloads simultaneously stress processors, memory, bandwidth, and
other components.  From an algorithmic perspective, resource constraints
interact with time in a particularly direct way: the scheduler must decide
\emph{which jobs should overlap}, since resources occupied by a job are
released only when that job completes.  Consequently, total work alone is not
sufficient to characterize a good schedule.

A natural lower bound captures this resource--time interaction.  Define
\[
    V_k=\sum_{i=1}^{n}s_i v_{i,k},
    \qquad
    V=\max_{k\in[d]}V_k.
\]
Resource \(k\) can supply only one unit of capacity per unit of time, and hence
every feasible schedule satisfies
\[
    \operatorname{OPT}\ge V.
\]
The central question of this paper is when this elementary volume bound is
also approximately achievable.  In general, resource demands can create
severe fragmentation: even when the total volume is small, the resource
profiles of simultaneously executing jobs may prevent the available
capacities from being efficiently utilized.

\paragraph{\textbf{What is known?}}
Our problem is closely related to bin packing problem, more generally to vector bin packing problem.
  
 {\em Bin
packing} has been studied since the work of
Johnson in the 70s (see, e.g.,~\cite{Johnson74,JohnsonDUGG74}
and surveys~\cite{GalambosW95,Sgall14} for later work).

For the vector bin packing (VBP) problem, we pack multi-dimensional vectors and the simple greedy algorithm
has a competitive ratio of $O(d)$~\cite{GareyGJY76}.
In the offline setting, this was improved to a
$O(1 + \epsilon d + \log (1/\epsilon))$-approximation in \cite{ChekuriKhanna2004}.
On the other hand, \cite{AzarCohenKamaraShepherd2013}, using ideas from \cite{DeanGV08},
showed that a $d^{1-\epsilon}$-approximation would imply $NP = ZPP$.
For fixed $d$ (i.e., super-polynomial in $d$ running times),
the approximation factor improves to
$\ln d + O(1)$~\cite{BansalEliasKhan2016,BansalCapraraSviridenko2009,ChekuriKhanna2004}.
To overcome the strong lower bound for arbitrary $d$,
Epstein~\cite{Epstein03a} initiated the study of vector bin packing with variable bin sizes.
In the online setting, the $O(d)$ bound was improved to $O(\log(d^{\frac{1}{B-1}}))$
for $B \geq 2$ in \cite{AzarCohenFiatRoytman2016,AzarCohenPanigrahi2018}, who also showed an information-theoretic
lower bound of $\Omega\left(d^{\frac{1}{B} - \epsilon}\right)$
(for any $\epsilon > 0$). 

For large enough $B$, the lower bound vanishes; for this case, \cite{AzarCohenFiatRoytman2016}
gave a $(1+\epsilon)\cdot e$-competitive algorithm if
$B = (1/\epsilon^2)\cdot \Omega(\log d)$ (for any $\epsilon > 0$).
If vectors are splittable, \cite{AzarCohenFiatRoytman2016} gave an $e$-competitive algorithm,
and a matching lower bound was shown in~\cite{AzarCohenRoytman2017}.  

More recently, Sandeep~\cite{Sandeep2021} established nearly matching dimension-dependent hardness results: for all sufficiently large constant \(d\), \(d\)-dimensional VBP has no asymptotic \(c\log d\)-approximation for some constant \(c>0\), unless \(\mathrm{P}=\mathrm{NP}\), and when \(d\) is part of the input, VS is hard to approximate within \(\Omega((\log d)^{1-\epsilon})\) under a stronger complexity assumption. On the algorithmic side, Kulik, Mnich, and Shachnai~\cite{KulikMnichShachnai2023} recently obtained an asymptotic \((1+\ln d-\chi(d)+\epsilon)\)-approximation for VBP, for a positive function \(\chi(d)\), and a \((4/3+\epsilon)\)-approximation for \(d=2\). These results highlight the intrinsic difficulty of multidimensional packing for unrestricted vectors.

The situation changes dramatically when individual vectors are
\emph{small}.  Azar et al.~\cite{AzarCohenFiatRoytman2016} initiated a systematic study
of this regime for online vector bin packing.  While arbitrary vectors admit
strong dimension-dependent lower bounds, they showed that sufficiently small
vectors admit a competitive ratio arbitrarily close to \(e\), independent of
\(d\), and gave a corresponding deterministic result with a slightly stronger
smallness assumption.  This demonstrated an important phenomenon:
multidimensional packing can become qualitatively easier when no individual
item consumes a significant fraction of any resource.

\paragraph{\textbf{What is missing?}}
The small-vector results above concern \emph{static packing}: an arriving
vector is assigned to a bin and remains part of the load of that bin.
Processing times play no role.  Conversely, resource-constrained scheduling
captures processing times, but its general approximation guarantees typically
depend on the number of resources and do not exploit the regime in which
every job has tiny demand in every resource.

Our goal is to understand whether the favorable ``small-vector phenomenon''
survives once time is introduced.  This is not an immediate consequence of
small-vector bin packing.  Assigning a vector to a bin creates one
\(d\)-dimensional feasibility constraint; assigning a job a start time creates
constraints throughout an entire interval.  A long job may overlap many short
jobs, different duration scales interact with one another, and an overload can
potentially occur at many different times and in any of the \(d\) resources.
Thus, techniques that control the load of a fixed bin do not directly control
the temporal pattern of overlaps.  

This distinction is especially apparent in the online setting: when job \(i\)
arrives, we must irrevocably decide its start time without knowing future jobs.

\paragraph{\textbf{Our results.}}
We show that the small-vector assumption remains remarkably powerful even
after introducing heterogeneous processing times.  Our first result is a
near-optimal offline algorithm.  We normalize the maximum processing time to
one. This is easily done offline and can also be done online since we use the doubling technique on the makespan. \footnotetext{Throughout the paper, all logarithms are base 2 unless otherwise specified.}

\begin{theorem}[Informal]
For every \(0<\epsilon<1/4\), if
\[
    \max_{i,k}v_{i,k}
    \le
    \frac{\epsilon^2}{24\log(d/\epsilon)},
\]
there is a randomized algorithm that constructs a schedule of expected
makespan
\[
    (1+6\epsilon)\operatorname{OPT}+3.
\]
Moreover, the algorithm can be implemented in \(O(dn\log n)\) time.  

\end{theorem}

Thus, in the small-vector regime, the schedule can be made asymptotically
arbitrarily close to optimal even when \(d\) is not a constant.  The running
time is also nearly linear in the input size.  This is in sharp contrast with
configuration-based approaches to general vector scheduling: for example, the
classical fixed-dimensional PTAS discretizes vectors, enumerates
configurations, and combines dynamic programming with LP rounding, leading to
a polynomial whose exponent grows rapidly with \(d\) and the accuracy
parameter.

We next consider the online problem.  Here, jobs arrive one at a time, and the
start time of a job must be fixed upon its arrival.
Let \(
     T=\frac{\max_i s_i}{\min_i s_i}.
 \)

\begin{theorem}[Informal]
 If
\(
    \max_{i,k}v_{i,k}
    \le
    \frac{1}{100\log d\log T},
\)
there is a randomized online algorithm whose expected makespan is at most
\(
    96\,\operatorname{OPT}.
\)
\end{theorem}

In particular, the competitive ratio is an absolute constant, independent of
both the dimension and the number of jobs; \(d\) and \(T\) appear only in the
smallness requirement.  We further remove the randomization.

\begin{theorem}[Informal]
For \(d\ge2\), there is a deterministic online algorithm with makespan at most
\(
    160\,\operatorname{OPT},
\)
provided that every coordinate satisfies
\[
    \max_{i,k}v_{i,k}
    \le
    \frac{1}{4000\log d\log T} .
\]

\end{theorem}

Taken together, these results extend the small-vector phenomenon from static
multidimensional packing to temporal resource scheduling: small vectors allow
a near-optimal offline approximation and constant-competitive online
algorithms even when jobs have heterogeneous processing times.

\subsection{Techniques}
Our offline algorithm is based on a simple random-placement principle.  
We first round processing times geometrically and consider the jobs from
longer to shorter processing times.  Let
\[
    V=\max_k\sum_i t_i v_{i,k}
\]
be the max-coordinate resource--time volume after rounding.  We create a
time horizon of length only slightly larger than \(V\), and independently
assign every job a random start time aligned with its rounded processing time.  If
the job fits throughout its proposed interval, we keep the assignment;
otherwise, we send the job to an \emph{overflow pool}.

The main difficulty is that a job can fail because of an overload at any time
in its execution interval.  A naive union bound over all possible times would
destroy the desired guarantee.  Geometric rounding gives us the crucial
structural property: for a job of a fixed processing time, an unsuccessful placement
can be witnessed at only a small set of representative time points.  We can
therefore apply concentration bounds at these points and union-bound over
both the relevant times and the \(d\) resource dimensions.  Small coordinates
ensure that the probability of overflow is tiny.

We then exploit an apparent mismatch in our favor.  The overflow pool can be
scheduled by a simple shelf algorithm whose worst-case guarantee loses a
factor depending on \(d\).  However, the expected resource--time volume that
reaches the overflow pool is smaller by a compensating factor.  Consequently,
the overflow contributes only an \(O(\epsilon)\) fraction of the original
volume, yielding the \((1+O(\epsilon))\) offline guarantee.  This
``random placement plus cheap repair'' viewpoint is substantially simpler
than explicitly enumerating multidimensional packing configurations.

The randomized online algorithm follows the same high-level principle, but
must address two additional obstacles.  First, the final resource--time volume
is not known in advance; we handle this using a standard doubling argument.
Second, an arriving job can interact with multiple processing-time scales.
After rounding processing times to powers of two, all relevant events lie on a
controlled temporal grid.  Concentration followed by a union bound over the
grid and the resource dimensions shows that each job enters the overflow pool
with sufficiently small probability.  A separate online shelf schedule then
absorbs the rejected jobs while losing only a constant factor.

Finally, our deterministic online algorithm replaces random start times with
a fractional schedule and an exponential potential function.  The fractional
schedule specifies a distribution over possible start times for every job.
When processing a job, we show that the expected potential under this
distribution does not increase; hence some start time in its support does not
increase the potential.  Choosing such a start time deterministically
maintains control over every resource and every relevant time point.
This converts the concentration argument underlying the randomized algorithm
into a deterministic online invariant.

\subsection{Other Related Work}

\paragraph{Vector scheduling.}
Vector scheduling has been studied extensively as models of jobs that simultaneously require several resources. Chekuri and Khanna~\cite{ChekuriKhanna2004} systematically studied \emph{vector scheduling} (VS), where \(d\)-dimensional vectors are assigned to a fixed number of machines so as to minimize the maximum load over all machines and coordinates. They gave a PTAS for VS when \(d\) is fixed and an \(O(\log^2 d)\)-approximation for arbitrary \(d\). 

The online version of vector scheduling has also received substantial attention. Meyerson, Roytman, and Tagiku~\cite{MeyersonRoytmanTagiku2013} obtained an \(O(\log d)\)-competitive algorithm, and Im, Kell, Kulkarni, and Panigrahi~\cite{ImKellKulkarniPanigrahi2015} subsequently established an asymptotically tight \(\Theta(\log d/\log\log d)\) dependence for deterministic online vector scheduling. Azar, Cohen, and Panigrahi~\cite{AzarCohenPanigrahi2018} showed that randomization does not asymptotically eliminate this dependence by proving an \(\Omega(\log d/\log\log d)\) lower bound for randomized algorithms. Thus, for unrestricted vectors, a dependence on the dimension is unavoidable even online. More recently, Gupta, Krishnaswamy, Sandeep, and Sundaresan~\cite{GuptaKrishnaswamySandeepSundaresan2023} studied fully dynamic vector scheduling and vector bin packing with recourse. For VBP, under the assumption that every coordinate is \(O(1/\log d)\), they obtain an \(O(1)\)-competitive algorithm with \(O(\log n)\) amortized recourse. Their result provides further evidence that small per-coordinate demands can fundamentally change the behavior of multidimensional packing. Recent work~\cite{das2025approximation} considered online multidimensional vector scheduling under precedence constraints and showed tight logarithmic competitive ratios.

\paragraph{Temporal and dynamic vector packing.}
Recent work has considered temporal variants of multidimensional packing. In dynamic vector bin packing, a job has a multidimensional resource demand together with an activity interval, and the algorithm decides \emph{where} to place the job among multiple bins or servers while respecting the capacity constraints of every active bin~\cite{vanBevernMelnikovSmirnovTsidulko2023,MurhekarArbourMaiRao2023}. In the model of van Bevern et al.~\cite{vanBevernMelnikovSmirnovTsidulko2023}, the start and end times of each request are part of the input and the objective is to minimize the number of bins. Murhekar et al.~\cite{MurhekarArbourMaiRao2023} instead consider online jobs that are assigned to servers upon arrival and depart after completing their processing; their objective is to minimize total server usage time.

Our problem differs in the central optimization decision. All jobs share a single \(d\)-dimensional capacity profile, and the scheduler decides \emph{when} each job executes. Choosing a start time for job \(i\) determines its activity interval \([f(i),f(i)+s_i)\), and hence which other jobs it overlaps. Thus, unlike the above dynamic packing models, the temporal overlap pattern itself is part of the optimization. The objective is to minimize makespan rather than the number or usage time of bins. This distinction is particularly important for heterogeneous processing times, since the scheduler must coordinate overlaps across multiple duration scales rather than assign already active jobs to different bins.

\section{$(1+\epsilon)$-Approximation for Offline Vector Scheduling}\label{vecsect}

We normalize the longest processing time of a vector to be 1.
Let $OPT$ denote the length of the optimal schedule when vector $v_i$ has processing time $s_i$. 

In this section, we prove \autoref{offlineapprox}, which states that for $d\geq2$ and any $\epsilon$ such that $\frac{1}{4}>\epsilon>0$ and $1/\epsilon\in\mathbb{Z}$, \autoref{vecSch} returns a schedule of expected makespan at most $(1+6\epsilon)OPT+3$ if each coordinate of every vector has size at most $\epsilon^2/(24\log(d/\epsilon))$.

\begin{theorem}\label{offlineapprox}
Let $d\geq2$ and let $\frac{1}{4}>\epsilon>0$ be such that $1/\epsilon\in\mathbb{Z}$.
There is an offline algorithm for vector scheduling that produces a solution of expected makespan at most $(1+6\epsilon)OPT+3$ for vectors whose coordinates are at most $\frac{\epsilon^2}{24\log(d/\epsilon)}$, where the maximum processing time of a vector is normalized to be 1.
\end{theorem}

Round all processing times up to the next power of $2^\epsilon$. 
Since $2^\epsilon\leq1+\epsilon$, this increases every processing time by a factor of at most $1+\epsilon$.

Let $t_i$ denote the processing time of $v_i$ after rounding.
Let $OPT'$ denote the length of the optimal schedule when vector $v_i$ has processing time $t_i$.
The optimum amount of time required to schedule these vectors increases by at most a factor of $(1+\epsilon)$, that is, $OPT'\leq(1+\epsilon)OPT$.

Recall that
$$
V=\max_{k\in[d]}\sum_{i=1}^n t_i v_{i,k}
$$
is the \emph{\volume} of our vectors.
One can see that $V$ is a lower bound for $OPT'$.

Let
\begin{equation}
\label{eqn:M}
M=\frac{V}{1-\epsilon}+1\leq(1+2\epsilon)V+1.
\end{equation}

\begin{algorithm}\caption{Vector scheduling algorithm}
\label{vecSch}

{\bf Input}: $\epsilon>0$, and a set of vectors $v_1,v_2,\ldots,v_n$ with processing times $t_1,t_2,\ldots,t_n$.

W.l.o.g., since the entire instance is known in advance, index the vectors so that $t_1\geq t_2\geq\cdots\geq t_n$.

For each $i=1,2,\ldots,n$, and for each discrete \stime $0\leq t<\frac{V}{1-\epsilon}$ such that $\frac{t}{t_i}\in\mathbb{Z}$, assign $t$ as the starting \stime of $v_i$ with probability $\dfrac{1}{\lceil(\frac{V}{1-\epsilon})/t_i\rceil}$, independently of the assignments of the other vectors. (That is, the starting \stime of $v_i$ is chosen uniformly from these \stimes.)

If a vector cannot fit on the strip when started at its assigned starting \stime, put it in an ``overflow pool''.

Round the processing time of each vector in the overflow pool up to the next largest power of 2.

Apply Algorithm \ref{shelf} to the vectors in the overflow pool (with rounded processing times), and append the resulting schedule after the schedule of the vectors that fit on the strip.
\end{algorithm}

\paragraph{Vector scheduling algorithm.}
Our vector scheduling algorithm (Algorithm \ref{vecSch}) works as follows.
Since the entire instance is known in advance, we may index the vectors so that $t_1\geq t_2\geq\ldots\geq t_n$.

For each $i=1,2,\ldots,n$, and for each \stime $0\leq t<\frac{V}{1-\epsilon}$ such that $\frac{t}{t_i}\in\mathbb{Z}$, assign $t$ as the starting \stime of $v_i$ with probability $\dfrac{1}{\lceil(\frac{V}{1-\epsilon})/t_i\rceil}$, independently of the assignments of the other vectors. As a consequence, the starting \stime of $v_i$ is chosen uniformly from these possible \stimes.

If a vector cannot fit on the strip when started at its assigned \stime, we put it in an ``overflow pool'', round its processing time to the next largest power of 2, and apply the Shelf Algorithm, \autoref{shelf}, to the vectors in the overflow pool. We append the resulting schedule after the schedule of the vectors that fit on the strip.

Remark: We note that it is straightforward to implement the algorithm in almost linear time (i.e. $O(dn\log n)$).

\begin{algorithm}\caption{Shelf algorithm}
\label{shelf}

{\bf Input}: A sequence of vectors $\Bar{v}_1,\Bar{v}_2,\ldots$ where $\Bar{v}_i$ has a processing time $t_i$ (which is a power of 2).

Replace each vector with a larger vector for which all coordinates are the same and are equal to the maximum coordinate $q_i=\max_k\Bar{v}_{i,k}\leq\sum_k\Bar{v}_{i,k}$.

For each modified vector, if there exists a shelf with processing time $t_i$ such that the vector fits there, then put the vector there. 
Otherwise open a new shelf with processing time $t_i$ and put the vector in the newly opened shelf.
\end{algorithm}

\begin{lemma}
Suppose that \autoref{vecSch} tentatively assigns a vector $v_i$ to \stime $t$. 
There is a set of \stimes $\tau_i(t)\subset[t,t+t_i)$ of size at most $2/\epsilon$ such that, if $v_i$ cannot fit at \stime $t$, then there is some \stime $t'\in\tau_i(t)$ such that the schedule is overloaded at \stime $t'$, that is, for some $k\in[d]$,
$$
\sum_{v_j\text{ already scheduled, }\,t'\in[f(j),f(j)+t_j)}v_{j,k}+v_{i,k}>1.
$$
\end{lemma}

\begin{proof}
Recall that each $t_j$ is equal to $2^{\epsilon w_j}$ for some $w_j\in\mathbb{Z}$. Let $g=1/\epsilon$, which is an integer. For each $j$, write $ w_j=q_jg+r_j,
$ where $q_j\in\mathbb{Z}$ and $0\leq r_j<g$. 
So, $ t_j=2^{q_j+r_j/g}. $

Define
$$
\tau_i(t):=
\{m2^{q_i+r'/g}:r'\in[g-1]\cup\{0\},\,m\in\mathbb{Z},\,
t\leq m2^{q_i+r'/g}<t+t_i\}.
$$

We first show that if $v_i$ cannot fit when started at \stime $t$, then the addition of $v_i$ causes an overload at some \stime in $\tau_i(t)$.

Suppose that $v_i$ cannot fit when started at \stime $t$. If the addition of $v_i$ already causes an overload at $t$, then $t\in\tau_i(t)$, since $t$ is an admissible starting \stime for $v_i$ and hence is a multiple of
$ t_i=2^{q_i+r_i/g}. $

Otherwise, let $t'>t$ be the first \stime in $[t,t+t_i)$ at which the addition of $v_i$ causes an overload. The load on the strip can change only at the starting and completion \stimes of previously scheduled vectors. Moreover, the completion of a vector can only decrease the load. Therefore, $t'$ must be the starting \stime of some previously scheduled vector $v_j$.

Since the vectors are considered in nonincreasing order of processing time, $t_j\geq t_i$. Hence $w_j\geq w_i$, which implies $q_j\geq q_i$. Since $t'$ is an admissible start time of $v_j$, for some $m\in\mathbb{Z}$,
$
t'=m t_j
=m2^{q_j+r_j/g}
=  (m2^{q_j-q_i}   )2^{q_i+r_j/g}.
$
Since $q_j-q_i$ is a nonnegative integer, $m2^{q_j-q_i}\in\mathbb{Z}$. Thus, $t'$ is a multiple of $2^{q_i+r_j/g}$, and therefore $t'\in\tau_i(t)$.

It remains to bound the size of $\tau_i(t)$. For every $r'\in[g-1]\cup\{0\}$, consecutive multiples of $2^{q_i+r'/g}$ are separated by at least $2^{q_i}$. On the other hand,
$ t_i=2^{q_i+r_i/g}<2^{q_i+1}. $
Therefore, the interval $[t,t+t_i)$ contains at most two multiples of $2^{q_i+r'/g}$ for each $r'$. Since there are $g=1/\epsilon$ possible values of $r'$, we obtain
$ |\tau_i(t)|\leq 2g=\frac{2}{\epsilon}. $
\end{proof}

\begin{lemma}
The probability that a vector $v_i$ is placed in the overflow pool is at most $2\epsilon/d^3$.
\end{lemma}

\begin{proof}
Fix a possible start time $t$ of $v_i$ and condition on $v_i$ being assigned \stime $t$. Since the start times of the vectors are assigned independently, this conditioning does not affect the distribution of the start times assigned to the vectors considered before $v_i$.

For each vector $v_j$ considered before $v_i$, each potential start time is assigned to $v_j$ with probability $\dfrac{1}{\lceil(\frac{V}{1-\epsilon})/t_j\rceil}$.

We fix the value of $k$.
Consider any particular \stime $t'\in\tau_i(t)$.
Let $a_j=v_{j,k}\cdot24\log(d/\epsilon)/\epsilon^2$.
Note that $0\leq a_j\leq1$.

Let $X_j$ be the indicator random variable that equals $1$ if and only if the start time assigned to vector $v_j$ is $t_j\lfloor t'/t_j\rfloor$. If this is not an admissible starting \stime for $v_j$, then $X_j=0$. Since the start times of the vectors are assigned independently, the random variables $X_j$ are independent.

Let $X=\sum_{j<i}a_jX_j$,
and let $\mu=\frac{24\log(d/\epsilon)}{\epsilon^2(1+\epsilon)}$.

Note that $X_j=1$ with probability at most $\dfrac{1}{\lceil(\frac{V}{1-\epsilon})/t_j\rceil}\leq\dfrac{(1-\epsilon)t_j}{V}$.
Substituting this into $X=\sum_{j<i}a_jX_j$ and simplifying, we deduce that
$$
\E[X]\leq\sum_{j<i}\frac{(1-\epsilon)t_j}{V}a_j
\leq\sum_{j<i}\frac{(1-\epsilon)t_j}{V}v_{j,k}\frac{24\log(d/\epsilon)}{\epsilon^2}
\leq\frac{24\log(d/\epsilon)}{\epsilon^2(1+\epsilon)}
=\mu,
$$
where the last inequality uses $1-\epsilon\leq1/(1+\epsilon)$.

Let $\eta=\frac{11\epsilon}{12}$.
Since each $0\leq a_jX_j\leq1$, we may think of $X$ as the sum of independent random variables that take values between 0 and 1.

Recall that the maximum value of each coordinate of any vector is at most $\frac{\epsilon^2}{24\log(d/\epsilon)}$. For the purpose of the analysis only, the coefficients $a_j$ scale the resource usage in dimension $k$ by the factor $24\log(d/\epsilon)/\epsilon^2$. This scaling does not change the vectors or the capacity of the scheduling problem; it only means that the random variable $X$ is a normalized representation of the resource usage of the vectors whose assigned start times cause them to overlap \stime $t'$.

Suppose that inserting $v_i$ causes an overflow in dimension $k$ at \stime $t'$. Then the previously scheduled vectors that overlap $t'$ use more than $1-v_{i,k}$ units of resource $k$. Since $v_{i,k}\leq\epsilon^2/(24\log(d/\epsilon))$, these vectors use more than $1-\epsilon^2/(24\log(d/\epsilon))$ units of resource $k$.

Every previously scheduled vector that overlaps $t'$ has an assigned start time that causes it to overlap $t'$, and hence its normalized resource usage is included in $X$. Moreover, $X$ may also include vectors whose assigned start times cause them to overlap $t'$ but which were subsequently placed in the overflow pool. Thus, if there is an overflow in dimension $k$ at \stime $t'$ when attempting to schedule $v_i$, then $X\geq\frac{24\log(d/\epsilon)}{\epsilon^2}-1$.

Under our assumptions on $\epsilon$ and $d$, $\frac{24\log(d/\epsilon)}{\epsilon^2}-1\geq(1+\eta)\mu$. Applying the Chernoff bound~\cite{Chernof}, we find that
$$
\Pr\left[X\geq\frac{24\log(d/\epsilon)}{\epsilon^2}-1\right]
\leq
\Pr[X\geq(1+\eta)\mu]
\leq
e^{-\frac{\eta^2\mu}{3}}
=
e^{-\frac{\epsilon^2\cdot11^2\cdot24\log(d/\epsilon)}
{12^2\cdot3\epsilon^2(1+\epsilon)}}
\leq
\frac{\epsilon^2}{d^4}.
$$

Thus, for a fixed $k\in[d]$ and a fixed \stime $t'\in\tau_i(t)$, the probability that the addition of $v_i$ causes an overload in dimension $k$ at \stime $t'$ is at most $\epsilon^2/d^4$.

Since this holds for each $k\in[d]$, the probability that the addition of $v_i$ causes an overload at \stime $t'$ is at most $\epsilon^2/d^3$.

Summing over all \stimes $t'\in\tau_i(t)$, and using $|\tau_i(t)|\leq2/\epsilon$, the probability that $v_i$ overflows, conditioned on being assigned start time $t$, is at most $2\epsilon/d^3$. Since this bound holds for every possible start time $t$ of $v_i$, the probability that $v_i$ is placed in the overflow pool is at most $2\epsilon/d^3$.
\end{proof}

Now we consider the makespan of the schedule returned by the shelf algorithm.

\begin{lemma}\label{shelfbound}
Let $S$ be the set of vectors given to the shelf algorithm. Suppose that each coordinate of a vector in $S$ has size at most $\delta\leq1/2$, that the processing time of each vector is a power of 2, and that the maximum processing time of a vector is normalized to be 1. Let
$\Bar{V}=\sum_{i\in S}t_i\max_{k\in[d]}v_{i,k}$ and let $\hat{V}=\max_{k\in[d]}\sum_{i\in S}t_iv_{i,k}$ be the \volume of the vectors in $S$. Then the makespan of the shelf algorithm is at most
$$
(1+2\delta)\Bar{V}+2
\leq
(1+2\delta)d\hat{V}+2.
$$
\end{lemma}

\begin{proof}
For each vector $v_i\in S$, let $q_i=\max_{k\in[d]}v_{i,k}$, and let $S'$ be the modified vectors obtained from $S$ by replacing each vector $v_i$ with a larger vector $\Bar{v}_i$ for which all coordinates are equal to $q_i$. Then
$$
V(S')=\sum_{i\in S}t_iq_i=\Bar{V}.
$$
Moreover, since $q_i\leq\sum_{k\in[d]}v_{i,k}$,
$$
\Bar{V}
\leq\sum_{k\in[d]}\sum_{i\in S}t_iv_{i,k}
\leq d\hat{V}.
$$

The schedule returned by the shelf algorithm consists of consecutive shelves, where each vector fits inside exactly one shelf. For each shelf length, there is at most one shelf that has more than $\delta$ unused capacity. Indeed, if two shelves of the same length each had more than $\delta$ unused capacity, then the first vector placed in the later shelf, whose size is at most $\delta$, would have fit in the earlier shelf.

Let $\tau_i$ denote the number of shelves of length $2^{-i}$.
The total \volume of vectors packed inside shelves of length $2^{-i}$ is at least $2^{-i}(1-\delta)(\tau_i-1)$.
This is because, except for at most one shelf of each length, every shelf has at least $1-\delta$ of its capacity filled.
Hence,
$$
\sum_{i\geq0}2^{-i}(1-\delta)(\tau_i-1)\leq\Bar{V}.
$$

Note that $1/(1-\delta)\leq(1+2\delta)$. By rearranging the terms we get that the total length is at most
$$
\sum_{i\geq0}2^{-i}\tau_i
\leq(1+2\delta)\Bar{V}+\sum_{i\geq0}2^{-i}
\leq(1+2\delta)\Bar{V}+2
\leq(1+2\delta)d\hat{V}+2.
$$
\end{proof}

\begin{proof} (of \autoref{offlineapprox})

We now show that \autoref{vecSch} produces a solution of makespan at most $(1+6\epsilon)OPT+3$, in expectation.

Let $S$ be the random set of vectors placed in the overflow pool, and let $q_i=\max_kv_{i,k}$. Since each vector overflows with probability at most $\frac{2\epsilon}{d^3}$,
$$
\E\left[\sum_{i\in S}t_iq_i\right]
=
\sum_i t_iq_i\Pr[i\in S]
\leq
\frac{2\epsilon}{d^3}\sum_i t_iq_i
\leq
\frac{2\epsilon}{d^2}V.
$$

Before applying the Shelf Algorithm, we round the processing time of each vector in the overflow pool to the next largest power of 2. This increases every such processing time by a factor of less than 2. Therefore, if $\Bar{V}_S$ denotes $\sum_{i\in S}t_iq_i$ after this rounding, then $\E[\Bar{V}_S]\leq\frac{4\epsilon}{d^2}V$.

In our application of \autoref{shelfbound}, take $\delta=\epsilon^2/(24\log(d/\epsilon))$. By \autoref{shelfbound}, the Shelf Algorithm adds an expected extra processing time of at most $(1+2\delta)\E[\Bar{V}_S]+2\leq(1+2\delta)\frac{4\epsilon}{d^2}V+2$. Since $\delta\leq1/2$, this is at most $\frac{8\epsilon}{d^2}V+2$.

Every vector that is not placed in the overflow pool starts before $\frac{V}{1-\epsilon}$ and has processing time at most 1. Thus, the makespan of this part of the schedule is at most $M$. Therefore, the expected length of the entire schedule is at most
$$
M+\frac{8\epsilon}{d^2}V+2
\leq(1+2\epsilon)V+1+2\epsilon V+2
\leq(1+4\epsilon)V+3,
$$
where the first inequality uses the bound from Equation~\eqref{eqn:M} and the fact that $d\geq2$.

Finally, $V\leq OPT'\leq(1+\epsilon)OPT$. Therefore,
$$
(1+4\epsilon)V+3
\leq
(1+4\epsilon)(1+\epsilon)OPT+3
\leq
(1+6\epsilon)OPT+3,
$$
where the last inequality follows from $\epsilon\leq1/4$.

\end{proof}

\section{Online with Smaller Vector Size}\label{onlinesmall}

In this section, we consider the online version of the vector scheduling problem given in the Introduction.

Let $T$ be the ratio of the maximum vector processing time to the minimum vector processing time (we round it to the next power of 2). As before, we normalize the maximum processing time of a vector to be 1. Hence the minimum processing time is at least $1/T$. We assume that $d,T\geq2$ and that all logarithms are base 2.

First, we round up the processing time of each vector to the next largest power of 2. 

Let $t_i$ denote the processing time of $v_i$ after rounding.
Let $OPT'$ denote the length of the optimal schedule when vector $v_i$ has processing time $t_i$.
The optimal makespan increases by at most a factor of 2, that is, $OPT'\leq2OPT$.

Let
$$
V=\max_{k\in[d]}\sum_{i=1}^n t_i v_{i,k}
$$
be the overall \emph{\volume} of our vectors after rounding the processing times.
One can see that $V$ is a lower bound for $OPT'$.

Let
$$
M=\left\lceil2V\right\rceil\leq2V+1.
$$

Our algorithm uses $M$ as a parameter. However, for computing $M$ we need to know the set of vectors in advance, which is not given to us. To overcome this, we use the classical folklore doubling method. Specifically, we maintain a guess for $M$ and start a new phase with a doubled guess whenever the current guess is no longer sufficient. As usual, this increases the competitive ratio by at most a factor of $4$. Thus, we first analyze the algorithm assuming that a sufficiently large value of $M$ is known.

\paragraph{\textbf{Online algorithm.}}

First, we define the algorithm for a certain choice of $\Bar{M}\in\mathbb{N}$.
This algorithm, \autoref{randalg1}, is given a choice of $\Bar{M}\in\mathbb{N}$ such that $\Bar{M}\geq\lceil2V\rceil$.

Since each $t_i$ is a power of 2 with $1/T\leq t_i\leq1$, and both $T$ and $\Bar{M}$ are integers, $\Bar{M}/t_i\in\mathbb{Z}$. Thus there are exactly $\Bar{M}/t_i$ \stimes $0\leq t<\Bar{M}$ such that $t/t_i\in\mathbb{Z}$.

When vector $v_i$ arrives, independently of all previous random choices, for each \stime $0\leq t<\Bar{M}$ such that $t/t_i\in\mathbb{Z}$, assign $t$ as the tentative starting \stime of $v_i$ with probability $t_i/\Bar{M}$. Equivalently, one of these \stimes is chosen uniformly at random.
We call the selected \stime the \emph{tentatively assigned} \stime of $v_i$.

If a vector cannot fit at its tentatively assigned \stime, we put it in an ``overflow pool''.

We use Algorithm \ref{shelf} (the Shelf algorithm) for the overflow pool, starting the schedule for the overflow pool at \stime $\Bar{M}$.
Since the starting \stime of the schedule for the overflow pool is known in advance and the Shelf algorithm can be run online, each vector in the overflow pool can be assigned a \stime as soon as it is placed in the overflow pool. 
Thus, the algorithm does not need to wait for future vectors to arrive before scheduling an overflow vector.

\begin{algorithm}
\caption{Vector schedule for given \volume.}
\label{randalg1}

{\bf Input}: $\Bar{M}$, and a sequence of vectors $v_1,v_2,\ldots,v_m$ of \volume $V$, with processing times $t_1,t_2,\ldots,t_m$, that have maximum coordinate at most $1/(100\log d\log T)$ and arrive online, such that $\Bar{M}\geq\lceil2V\rceil$.

When vector $v_i$ arrives, independently of all previous random choices, for each \stime $0\leq t<\Bar{M}$ such that $t/t_i\in\mathbb{Z}$, assign $t$ as the tentative starting \stime of $v_i$ with probability $t_i/\Bar{M}$. Equivalently, one of these \stimes is chosen uniformly at random.

If $v_i$ cannot fit at its tentatively assigned \stime, put it in an ``overflow pool''.

Use Algorithm \ref{shelf} (the Shelf algorithm) for the overflow pool, starting the schedule for the overflow pool at \stime $\Bar{M}$.

\end{algorithm}

First, we show in the next lemma that the probability that a vector is placed in the overflow pool is at most $1/d^3$.

\begin{lemma}
The probability that a vector $v_i$ is placed in the overflow pool by \autoref{randalg1} is at most $1/d^3$.
\end{lemma}

\begin{proof}
Suppose vector $v_i$ is tentatively assigned to \stime $\tau_i$. We condition on this assignment throughout the proof.

Fix $k\in[d]$ and let $\tau'\in[\tau_i,\tau_i+t_i)$. We first show that the probability that the addition of $v_i$ causes more than $1$ unit of resource $k$ to be used at \stime $\tau'$ is at most $1/(Td^4)$.

Let $a_j=100v_{j,k}\log d\log T$.
Note that $0\leq a_j\leq1$.

Let $X_j$ be the indicator random variable that equals $1$ if and only if the \emph{tentatively assigned} \stime of vector $v_j$ is $t_j\lfloor\tau'/t_j\rfloor$. Since $\tau'<\Bar{M}$, this is an admissible \stime for $v_j$. Since the tentative \stimes of different vectors are chosen independently, the random variables $X_j$ are independent.

Let $X=\sum_{j<i}a_jX_j$ and let $\mu=50\log d\log T$.

Note that $X_j=1$ with probability $t_j/\Bar{M}$. Substituting this into $X=\sum_{j<i}a_jX_j$ and simplifying, we deduce that
\[
\E[X]
=\sum_{j<i}\frac{t_j}{\Bar{M}}a_j
\leq \sum_{j<i}\frac{t_j}{\Bar{M}}100v_{j,k}\log d\log T
\leq \frac{V}{\Bar{M}}100\log d\log T
\leq50\log d\log T
=\mu.
\]

Recall that the maximum value of each coordinate of any vector is at most $1/(100\log d\log T)$. For the purpose of the analysis only, the coefficients $a_j$ scale each resource usage by the factor $100\log d\log T$. This scaling does not change the vectors or the capacity of the scheduling problem; it only means that the random variable $X$ is a normalized representation of the resource usage of the tentatively assigned vectors that overlap \stime $\tau'$.

Suppose that inserting $v_i$ causes an overflow in dimension $k$ at \stime $\tau'$. Then the vectors that did not overflow and overlap \stime $\tau'$ use more than $1-v_{i,k}$ units of resource $k$. Since $v_{i,k}\leq1/(100\log d\log T)$, these vectors use more than $1-1/(100\log d\log T)$ units of resource $k$.

Every vector that did not overflow and overlaps \stime $\tau'$ was tentatively assigned to a \stime that causes it to overlap $\tau'$. The random variable $X$ includes the normalized resource usage of all such vectors, as well as possibly additional vectors that were tentatively assigned there but subsequently overflowed. Hence $X$ is at least the normalized resource usage of the vectors that did not overflow and overlap $\tau'$. Multiplying the preceding lower bound by $100\log d\log T$, we conclude that an overflow in dimension $k$ at \stime $\tau'$ implies $X\geq100\log d\log T-1$.

Recall the Chernoff bound which states that, for independent random variables taking values in $[0,1]$ and with $\E[X]\leq\mu$, we have $\Pr[X\geq(1+\delta)\mu]\leq e^{-\delta^2\mu/3}$.

Apply the Chernoff bound with $\delta=1/2$. Recall that $\mu=50\log d\log T$. 
Under our assumptions on $d$ and $T$, $100\log d\log T-1\geq75\log d\log T=(3/2)\mu$. 
Therefore
\[
\Pr[X\geq100\log d\log T-1]
\leq \Pr[X\geq(3/2)\mu]
\leq e^{-\frac{(1/4)\cdot50\log d\log T}{3}}
\leq \frac{1}{Td^4}.
\]

Thus, for any fixed $\tau'\in[\tau_i,\tau_i+t_i)$, the probability that the addition of $v_i$ causes an overflow in dimension $k$ at \stime $\tau'$ is at most $1/(Td^4)$.

We now consider all \stimes in $[\tau_i,\tau_i+t_i)$. Every rounded processing time $t_j$ is an integer multiple of $1/T$, and every tentative \stime of $v_j$ is a multiple of $t_j$. Therefore, both the starting and completion \stimes of every vector lie on the $1/T$ grid. Consequently, the resource usage is constant on every half-open interval $[r/T,(r+1)/T)$, where $r\in\mathbb{Z}$.

Hence, if an overflow occurs at some $\tau'\in[\tau_i,\tau_i+t_i)$, then it also occurs at the grid point $\frac{1}{T}\lfloor T\tau'\rfloor$. Since $t_i\leq1$, the half-open interval $[\tau_i,\tau_i+t_i)$ contains at most $T$ relevant grid points. Taking a union bound over these grid points, the probability that the addition of $v_i$ causes resource $k$ to exceed $1$ at some \stime in $[\tau_i,\tau_i+t_i)$ is at most $T\cdot(1/(Td^4))=1/d^4$.

Since this holds for each $k\in[d]$, taking a union bound over the $d$ coordinates implies that $v_i$ does not fit at its tentatively assigned \stime $\tau_i$ with probability at most $d\cdot(1/d^4)=1/d^3$.

Since the bound holds for every possible tentatively assigned \stime $\tau_i$, the probability that $v_i$ is placed in the overflow pool is at most $1/d^3$.
\end{proof}

\begin{lemma}\label{knownvolume}
Suppose Algorithm \ref{randalg1} is given a sequence of vectors $v_1,v_2,\ldots,v_n$ with processing times $t_1,t_2,\ldots,t_n$ and maximum coordinate at most $1/(100\log d\log T)$, and $\Bar{M}\in\mathbb{Z}$ such that $\Bar{M}\geq\lceil2V\rceil$.
Then Algorithm \ref{randalg1} produces a schedule for $v_1,v_2,\ldots,v_n$ of expected length at most $4\Bar{M}$.
\end{lemma}

\begin{proof}
Let $S$ be the set of vectors placed in the overflow pool and let $q_i=\max_kv_{i,k}$. Since each vector is placed in the overflow pool with probability at most $1/d^3$,
\[
\E\left[\sum_{i\in S}t_iq_i\right]
=\sum_i t_iq_i\Pr[i\in S]
\leq\frac{1}{d^3}\sum_i t_iq_i
\leq\frac{1}{d^3}\sum_{k\in[d]}\sum_i t_iv_{i,k}
\leq\frac{V}{d^2}.
\]

One can see that \autoref{shelfbound} also applies when the Shelf Algorithm is run online. Taking $\Bar{V}=\sum_{i\in S}t_iq_i$ and $\delta=1/(100\log d\log T)$, \autoref{shelfbound} gives a schedule of length at most $(1+2\delta)\Bar{V}+2$. Since $\delta\leq1/2$, the expected additional processing time of the Shelf Algorithm is at most $2V/d^2+2$.

Every vector that is not placed in the overflow pool completes by \stime $\Bar{M}$. Thus, the expected length of the entire schedule is at most
$$
\Bar{M}+\frac{2}{d^2}V+2\leq4\Bar{M}.
$$
\end{proof}

Now we prove the main theorem of this section.

\begin{theorem}\label{online96vec}
There is a randomized online algorithm that returns a vector schedule of expected length at most $96OPT$, where $OPT$ is the length of the optimal vector schedule before rounding the vector processing times, when the maximum coordinate of a vector is at most $1/(100\log d\log T)$.
\end{theorem}

\begin{proof}
If $M$ is known, then we can run Algorithm \ref{randalg1} using $\Bar{M}=M$ and obtain a schedule of expected length at most $4M$.
Recall that $M\leq3OPT'$ and $OPT'\leq2OPT$, so Algorithm \ref{randalg1} obtains a schedule of expected length at most $24OPT$.

Using the standard folklore doubling method described previously to guess $M$ increases the competitive ratio by at most a factor of $4$.
Therefore, the resulting online algorithm has expected length at most $96OPT$.
This completes the proof.
\end{proof}

\section{Derandomizing the Algorithm}
\label{sec:derandomization}

In this section, we derandomize the algorithm from the preceding section. We use the rounded processing times defined there.

For $m\in\mathbb{N}$, let a collection $f_0,f_1,\ldots,f_m$ be a
\emph{fractional vector schedule} for vectors $v_1,\ldots,v_n$.
Informally, $f_j(v_i)$ represents the fraction of vector $v_i$ assigned to start at \stime $j/T$. 
Formally, we require $f_j(v_i)\geq 0$ for every $i$ and $j$, 
\begin{equation*}
    \sum_{j=0}^{m} f_j(v_i)=1
    \qquad\text{for every }i\in[n],
\end{equation*}
and the resource constraints are satisfied at every \stime. 
That is,
\begin{equation}
\label{eq:fractional-capacity}
    \sum_{i=1}^{n}
    \sum_{\substack{j\in\{0,\ldots,m\}:\\ j/T\leq t<j/T+t_i}}
        f_j(v_i)v_{i,k}
    \leq 1
\end{equation}
for every $t\in\mathbb{R}_{+}$ and $k\in[d]$. 

The length of the fractional schedule is
\begin{equation*}
    \max_{\substack{i\in[n],\,j\in\{0,\ldots,m\}:\\f_j(v_i)>0}}
        \left(\frac{j}{T}+t_i\right).
\end{equation*}

The fraction $f_j(v_i)$ consumes $f_j(v_i)v_{i,k}$ units of resource $k$ throughout the half-open interval $[j/T,j/T+t_i)$. 
Thus, \eqref{eq:fractional-capacity} requires that, for every resource $k$ and \stime $t$, the total consumption of resource $k$ by all vector fractions active at \stime $t$ does not exceed its unit capacity.
For notational convenience we adopt the convention that $ f_j(v_i) =0 $ for $j\notin\{0,1,2,\dots,m\}$.

We say a fractional vector schedule $f_0,f_1,\dots,f_m$ has \defn{$1/2$-slack} if
$\sum_{i=1}^{n}
\sum_{\substack{j\in\{0,\ldots,m\}:\\j/T\leq t<j/T+t_i}}
f_j(v_i)v_{i,k}\leq\frac{1}{2}$
holds for any \stime $t$ and dimension $k$.

We now describe an algorithm to produce a vector schedule from a fractional vector schedule if the \volume is known.
Define $Q=4000 \log d\log T$.
We scale every coordinate of $v_i$ by $Q$ and, for notational simplicity, continue to
write $v_i$ for the scaled vector. We also scale the capacity of every resource from $1$ to $Q$, so this scaling does not change feasibility. 
 
This scaling convention is used in both the description and the analysis of the deterministic algorithm below.
The statement of Lemma~\ref{derand} refers to the original, unscaled vectors.

For a tentative start \stime $f'(v_i)$ for each vector $v_i$, define the
load after processing $v_1,\ldots,v_i$ by
\begin{equation*}
    L^i_{t,k}
      =\sum_{\substack{i'\leq i:\\f'(v_{i'})\leq t<f'(v_{i'})+t_{i'}}}
        v_{i',k}.
\end{equation*}
We also define
\begin{equation*}
    \alpha=e^{1/16},
    \qquad g(y)=\alpha^y.
\end{equation*}
The associated potential is
\begin{equation}
\label{eq:local-potential}
\Phi^i_{t,k}
=
g\left(
L^i_{t,k}
-
\alpha
\sum_{i' \leq i}
\sum_{j=T(t-t_{i'})+1}^{Tt}
f_j(v_{i'})v_{i',k}
\right).
\end{equation}
For a set of \stimes $A$, let
\begin{equation*}
    \Phi^i_{A,k}=\sum_{t\in A}\Phi^i_{t,k},
    \qquad
    \Phi^i_A=\sum_{t\in A}\sum_{k\in[d]}\Phi^i_{t,k}.
\end{equation*}
For $n'\in\mathbb{N}$ and $r'>0$, denote $[0,r']_{\frac{1}{n'}}=[0,r']\cap\frac{1}{n'}\mathbb{Z}$.

We now describe the algorithm for \autoref{derand}.
When $v_i$ arrives, the algorithm tentatively assigns it to a starting \stime 
$f'(v_i)=j/T$ with $f_j(v_i)>0$ that minimizes $\Phi^i_{[0,q-\frac{1}{T}]_{\frac{1}{T}}}$. 
If $v_i$ fits at that starting \stime, the assignment is made permanent; otherwise, $v_i$ is placed in an overflow pool.
We run the Shelf Algorithm, \autoref{shelf}, on the overflow pool, starting the schedule for the overflow pool at \stime $q$.
Since the starting \stime of the schedule for the overflow pool is known in advance and the Shelf Algorithm can be run online, each vector can be assigned a \stime as soon as it is placed in the overflow pool.

\begin{algorithm}
\caption{Deterministic Vector schedule for given fractional vector schedule and \volume.}
\label{fracalgVol}
\textbf{Input:} Vectors $v_1,\ldots,v_n$, fractional vector schedule $f_0,f_1,\ldots,f_m$ of length at most $q$.
The algorithm maintains that when $v_i$ arrives, tentative starting \stimes $f'(v_j)$ are assigned for all vectors $v_j$ with $j<i$.

When $v_i$ arrives, tentatively assign it to a starting \stime 
$j/T$ with $f_j(v_i)>0$ that minimizes $\Phi^i_{[0,q-\frac{1}{T}]_{\frac{1}{T}}}$. 
Denote this tentative starting \stime by $f'(v_i)$.

If $v_i$ fits at that \stime $f'(v_i)$, then assign it to that \stime. Otherwise, $v_i$ is placed in an overflow pool.

We use Algorithm \ref{shelf} (the Shelf algorithm) for the overflow pool, starting the schedule for the overflow pool at \stime $q$.

\textbf{Output:} Schedule for $v_1,\ldots,v_n$.
\end{algorithm}

We will prove the following lemma.
\begin{lemma}\label{derand}  
Suppose $d\geq 2$, every vector $v_i$ satisfies
$v_{i,k}\geq\frac{\max_{k'}v_{i,k'}}{d^2}$ for every $k\in[d]$,
and the maximum coordinate of every vector is at most $\frac{1}{4000 \log d\log T}$.
Let $q$ be an upper bound for the length of a $1/2$-slack fractional vector schedule, where $Tq\in\mathbb{Z}$.
Then \autoref{fracalgVol} deterministically constructs an online schedule of length
$q+2+\frac{4q}{d}$.
\end{lemma}

To prove Lemma~\ref{derand}, we first show that
$\Phi^i_{[0,q-\frac{1}{T}]_{\frac{1}{T}}}$ is nonincreasing in $i$. 
\begin{lemma}
\label{lem:potential-monotonicity}
For every $i\in[n]$,
\begin{equation*}
    \Phi^i_{[0,q-\frac{1}{T}]_{\frac{1}{T}}}
    \leq \Phi^{i-1}_{[0,q-\frac{1}{T}]_{\frac{1}{T}}}.
\end{equation*}
\end{lemma}

\begin{proof}
Fix $i$, $t$, and $k$. Consider the randomized assignment that starts $v_i$
at $j/T$ with probability $f_j(v_i)$. 
Conditional on the preceding
assignments, vector $v_i$ contributes $v_{i,k}$ to the load at \stime $t$ with
probability $\sum_{j=T(t-t_i)+1}^{Tt}f_j(v_i)$. It follows from
\eqref{eq:local-potential} that
\begin{align*}
    \mathbb{E}[ \Phi^i_{t,k} ]  
 & =   \left(\sum_{j=T(t-t_i)+1}^{Tt} f_j(v_i)\right) g \left(L^{i-1}_{t,k} + v_{i,k} - \alpha\sum_{i'\leq i}\sum_{j=T(t-t_{i'})+1}^{Tt}f_j(v_{i'})v_{i',k}\right) \nonumber\\
&\quad +\left(1-\sum_{j=T(t-t_i)+1}^{Tt}f_j(v_i)\right)
g\left(L^{i-1}_{t,k}-\alpha\sum_{i'\leq i}\sum_{j=T(t-t_{i'})+1}^{Tt}f_j(v_{i'})v_{i',k}\right)\nonumber\\
    &=\Phi^{i-1}_{t,k} 
      \alpha^{-\alpha v_{i,k}\sum_{j=T(t-t_i)+1}^{Tt}f_j(v_i)}
      \left(1+
      \left(\sum_{j=T(t-t_i)+1}^{Tt}f_j(v_i)\right)
      (\alpha^{v_{i,k}}-1)\right) \\
    &\leq \Phi^{i-1}_{t,k}\,
      \exp\!\left(
        -\alpha v_{i,k}\ln\alpha
          \sum_{j=T(t-t_i)+1}^{Tt}f_j(v_i)
        +v_{i,k}(\alpha-1)
          \sum_{j=T(t-t_i)+1}^{Tt}f_j(v_i)
      \right) \\
    &\leq \Phi^{i-1}_{t,k}.
\end{align*}
The first inequality uses
$\alpha^x-1\leq x(\alpha-1)$ for $x\in[0,1]$ and $1+x\leq e^x$; here $v_{i,k}\in[0,1]$ after scaling.
The second uses $\alpha\ln\alpha\geq\alpha-1$ for $\alpha\geq1$.
Summing over $t\in[0,q-\frac{1}{T}]_{\frac{1}{T}}$ and $k\in[d]$ gives
\begin{equation*}
    \mathbb{E}[\Phi^i_{[0,q-\frac{1}{T}]_{\frac{1}{T}}}]
      \leq \Phi^{i-1}_{[0,q-\frac{1}{T}]_{\frac{1}{T}}}.
\end{equation*}
Therefore, at least one starting \stime in the support of $f(v_i)$ does not increase the total potential. 
The algorithm chooses such a starting \stime, so the claim follows.
\end{proof}

We claim the following lemma holds.

\begin{lemma}
\label{eqneedprove}
If $L^i_{t,k}>Q-1$, then
\begin{equation*}
    L^i_{t,k}-\frac{3Q}{4}\geq96.
\end{equation*}
\end{lemma}

\begin{proof}
Because $L^i_{t,k}>Q-1$, it suffices to show that
\begin{equation}
\label{eq:Q-gap} 
    \frac{Q}{4}-1\geq96.
\end{equation}
The definition of $Q$ now implies the desired inequality under the standing assumptions $d,T\geq2$.
\end{proof}

We next prove the following lemma relating $\Phi^i_{t,k}$ and $L^i_{t,k}-Q+1$.

\begin{lemma}
\label{lem:potential-dominates-overload}
For every $i$, $t$, and $k$,
\begin{equation*}
    \Phi^i_{t,k}
      \geq \frac{1}{2}Td^6
        \bigl(L^i_{t,k}-Q+1\bigr)_+.
\end{equation*}
\end{lemma}

\begin{proof}
The claim is immediate when $L^i_{t,k}\leq Q-1$. Otherwise, the fractional
load bound and \eqref{eq:local-potential} yield
\begin{align*}
    \Phi^i_{t,k}
    &\geq \alpha^{L^i_{t,k}-\alpha \frac{2Q}{3}} 
    \geq \alpha^{L^i_{t,k}-\frac{3Q}{4}+\frac{Q}{40}} 
    \geq Td^6\,
       \alpha^{L^i_{t,k}-\frac{3Q}{4}} 
    \geq \frac{1}{2}Td^6
       \left(L^i_{t,k}-\frac{3Q}{4}\right) 
    \geq \frac{1}{2}Td^6
       \bigl(L^i_{t,k}-Q+1\bigr).
\end{align*}

For the first inequality, the original fractional schedule has $1/2$-slack, so after scaling by $Q$ its fractional load is at most $Q/2$ and hence at most $2Q/3$. Thus,
$\sum_{i'\leq i}
\sum_{\substack{j\in\{0,\ldots,m\}:\\j/T\leq t<j/T+t_{i'}}}
f_j(v_{i'})v_{i',k}\leq\frac{2Q}{3}$.
The second line uses $e^{1/16}\leq87/80$.
The third inequality follows from
\begin{equation*}
    \alpha^{Q/40}
      \geq e^{Q/640}
      \geq e^{6\log d\log T}
      \geq Td^6.
\end{equation*}
The fourth line follows from Lemma~\ref{eqneedprove} and the inequality
$e^{x/16}\geq x/2$ for $x\geq96$. The final line follows
from \eqref{eq:Q-gap}.
\end{proof}

We now prove \autoref{derand}.  
\begin{proof}(of \autoref{derand})

Because $L^0_{t,k}=0$ and the fractional prefix load is also zero,
$\Phi^0_{t,k}=1$. Hence
\begin{equation}
\label{eq:initial-potential}
    \Phi^i_{[0,q-\frac{1}{T}]_{\frac{1}{T}}}
      \leq \Phi^0_{[0,q-\frac{1}{T}]_{\frac{1}{T}}}
      =Tqd.
\end{equation}

Let $I^S$ be the set of vectors that do not fit at their tentative start
times, and let $V^S$ denote their \volume. 
Here, $V^S$ refers to the \volume of the overflow vectors after their coordinates have been scaled by $Q$.
For each $i\in I^S$, choose a resource $k(i)$ and
a \stime $t(i)$ for which assigning $v_i$ would violate the capacity constraint for resource $k(i)$ at \stime $t(i)$. Since all tentative starting and completion \stimes lie on the $\frac{1}{T}$-grid and the load is constant between consecutive grid points, we may choose such a \stime $t(i)\in\frac{1}{T}\mathbb{Z}$.
We say that $v_i$ is \emph{charged to} $(t(i),k(i))$.

For each fixed $t$ and $k$, suppose that at least one vector in $I^S$ is charged to $(t,k)$, and let $v_i$ be the first such vector considered.
Since $v_i$ does not fit at its tentative starting \stime and its scaled coordinates are at most $1$, immediately before $v_i$ is considered, the previously scheduled vectors use more than $Q-1$ units of resource $k$ at \stime $t$.

Moreover, every vector charged to $(t,k)$ has a tentative assignment that covers \stime $t$, and hence contributes its $k$-th coordinate to $L^n_{t,k}$. This remains true even if the vector is subsequently placed in the overflow pool, since $L^n_{t,k}$ is defined using the tentative
assignments. Therefore,
\begin{equation*}
\sum_{\substack{i\in I^S:\\t(i)=t,\;k(i)=k}}v_{i,k}
\leq
\left(L^n_{t,k}-(Q-1)\right)_+.
\end{equation*}

Since $v_{i,k}\geq\frac{\max_{k'\in[d]}v_{i,k'}}{d^2}$ and $t_i\leq1$ for all $i$, we have
\begin{equation}\label{VSeq}
\begin{aligned}
V^S
&= \max_{k\in[d]}\sum_{i\in I^S}v_{i,k}t_i \\
&\leq \sum_{i\in I^S}\sum_{k\in[d]}v_{i,k}t_i \\
&= \sum_{\substack{0\leq t \leq q-\frac{1}{T}:\\ t\in\frac{1}{T}\mathbb{Z}}}
   \sum_{k\in[d]}
   \sum_{\substack{i\in I^S:\\ t(i)=t,\;k(i)=k}}
   \sum_{k'\in[d]}v_{i,k'}t_i \\
&\leq \sum_{\substack{0\leq t\leq q-\frac{1}{T}:\\ t\in\frac{1}{T}\mathbb{Z}}}
   \sum_{k\in[d]}
   \sum_{\substack{i\in I^S:\\ t(i)=t,\;k(i)=k}}
   d^3v_{i,k} \\
&\leq d^3
   \sum_{\substack{0\leq t\leq q-\frac{1}{T}:\\ t\in\frac{1}{T}\mathbb{Z}}}
   \sum_{k\in[d]}
   \left(L^n_{t,k}-(Q-1)\right)_+ .
\end{aligned}
\end{equation}
Here, $(x)_+=\max\{0,x\}$.

We show that \autoref{fracalgVol} deterministically produces a vector schedule of length at most $q+2+\frac{4q}{d}$.
First note that if the vector $v_i$ is not in the overflow pool, then it is placed at a spot $j$ such that $f_j(v_i)>0$.  
By definition, $\frac{j}{T}+t_i\leq q$.
In other words, the vectors not in the overflow pool all finish before \stime $q$.

Combining \eqref{VSeq},
Lemma~\ref{lem:potential-dominates-overload}, and
Equation~\eqref{eq:initial-potential}, we obtain the following bound on the \volume of the overflow vectors.
\begin{align*}
    V^S
    &\leq d^3
      \sum_{\substack{0\leq t\leq q-\frac{1}{T}:\\ t\in\frac{1}{T}\mathbb{Z}}}
      \sum_{k\in[d]}
        \bigl(L^n_{t,k}-Q+1\bigr)_+ 
      \leq \frac{2d^3}{Td^6}
      \sum_{\substack{0\leq t\leq q-\frac{1}{T}:\\ t\in\frac{1}{T}\mathbb{Z}}}
      \sum_{k\in[d]}\Phi^n_{t,k} 
    \leq \frac{2q}{d^2}.
\label{overflowbd}
\end{align*}

We now return to the original, unscaled vectors. The \volume of the overflow vectors is therefore $\frac{V^S}{Q}$, and their maximum coordinate is at most $\frac{1}{Q}$. 
By \autoref{shelfbound}, the length of the schedule returned by the Shelf algorithm is at most
\begin{equation*}
2+\left(1+\frac{2}{Q}\right)d\frac{V^S}{Q}
\leq2+2dV^S
\leq2+\frac{4q}{d}.
\end{equation*}

Thus the total length of the schedule returned by \autoref{fracalgVol} is at most $q+2+\frac{4q}{d}$.
    
\end{proof}

\begin{algorithm}
\caption{Fractional schedule for a given {\volume}}
\label{fracalg1}
\textbf{Input:}  
A value $\Bar{M}\in\mathbb{Z}_+$ and online vectors
$v_1,\ldots,v_n$ with processing times $t_1,\ldots,t_n$, {\volume} $V$, and
maximum coordinate at most $1/(100\log d\log T)$, where
\begin{equation*}
    \Bar{M}\geq\lceil2V\rceil.
\end{equation*}

When $v_i$ arrives, set
\begin{equation*}
    f_{Tt}(v_i)=\frac{t_i}{\Bar{M}}
\end{equation*}
for each $t\in[0,\Bar{M})$ such that $t/t_i\in\mathbb{Z}$, and set
$f_j(v_i)=0$ at every other $j\in\mathbb{Z}$.

\textbf{Output:} A resulting $1/2$-slack fractional vector schedule.
\end{algorithm}

To obtain the deterministic algorithm, we construct a fractional schedule following the ideas of Section~\ref{onlinesmall}. 
Instead of randomly choosing a starting \stime for $\widetilde v_i$, we assign equal fractional mass to the admissible multiples of $t_i$ in
$[0,\Bar{M})$. This construction is summarized in \autoref{fracalg1}.

Now we show how to obtain the deterministic algorithm.
\begin{theorem}
\label{thm:deterministic-online}
For $ d \geq 2$, there is a deterministic online algorithm that returns a vector schedule of
length at most $160 \,OPT$, where $OPT$ is the length of an
optimal schedule before processing-time rounding, provided that the maximum
coordinate of every input vector is at most
\begin{equation*}
    \frac{1}{Q}
      =\frac{1}{4000 \log d\log T}.
\end{equation*}
\end{theorem}
\begin{proof} 
First we modify the received set of vectors $v_1,v_2,\ldots,v_n$ so that no coordinate of any vector $v_i$ is less than $\frac{1}{d^2}$ times the maximum coordinate.
That is, for each coordinate $v_{ik}$ that is less than
$\frac{\max_{k'}v_{i,k'}}{d^2}$, set that coordinate to
$\frac{\max_{k'}v_{i,k'}}{d^2}$.
Call this new set of vectors $\tilde{v_1},\tilde{v_2},\ldots,\tilde{v_n}$.
Let $\tilde{V}$ be the \volume of these vectors.
\begin{lemma}
    $\tilde{V}\leq\left(1+\frac{1}{d}\right)V$. 
\end{lemma}
\begin{proof}
For each $i\in[n]$, let $k(i)$ be such that $v_{i,k(i)}$ is a maximum coordinate of $v_i$.  

\[
\begin{aligned}
\tilde{V}
& = \max_{k\in[d]} \sum_{i=1}^{n}\tilde{v}_{i,k}t_i
 \leq \max_{k\in[d]} \sum_{i=1}^{n}
\left(v_{i,k}+\frac{1}{d^2}v_{i,k(i)}\right)t_i \\[2pt]
&\leq \max_{k\in[d]}\sum_{i=1}^{n}v_{i,k}t_i
+\frac{1}{d^2}\sum_{i'=1}^{n}v_{i',k(i')}t_{i'}
= V+\frac{1}{d^2}\sum_{k'\in[d]}
\sum_{\substack{i'\in[n]:\\k(i')=k'}}
v_{i',k'}t_{i'} \\[2pt]
&\leq V+\frac{1}{d^2}\sum_{k'\in[d]}\sum_{i'\in[n]}v_{i',k'}t_{i'}
\leq V+\frac{1}{d^2}\sum_{k'\in[d]}V
\leq V+\frac{1}{d}V.
\end{aligned}
\]
\end{proof}

The case $\widetilde V=0$ is trivial, so assume $\widetilde V>0$.
Let $M=\lceil2\widetilde V\rceil$. First suppose that $M$ is known and set $\Bar{M}=M$.
We apply \autoref{fracalg1} to the modified vectors.

Since $\Bar{M}$ is an integer and every rounded processing time $t_i$ is a power of $2$ between $\frac{1}{T}$ and $1$, $\Bar{M}/t_i\in\mathbb{Z}$. Thus, the fractional masses assigned to each vector sum to $1$. Moreover, the latest starting \stime in the support of $f(\widetilde v_i)$ is $\Bar{M}-t_i$, so every vector fraction finishes by \stime $\Bar{M}$. Therefore, the fractional schedule has length at most $\Bar{M}$.

For every \stime $t\in[0,\Bar{M})$ and vector $\widetilde v_i$, exactly one start \stime in the support of
$f(\widetilde v_i)$ covers $t$. Therefore, for every $k\in[d]$,
\begin{equation*}
    \sum_{i=1}^{n}
    \sum_{\substack{j:\\j/T\leq t<j/T+t_i}}
        f_j(\widetilde v_i)\widetilde v_{i,k}
      =\sum_{i=1}^{n}\frac{t_i}{\Bar{M}}\widetilde v_{i,k}
      \leq\frac{\widetilde V}{\Bar{M}}
      \leq\frac{1}{2}.
\end{equation*}

By construction, the modified vectors satisfy
$\widetilde v_{i,k}\geq\frac{\max_{k'}\widetilde v_{i,k'}}{d^2}$, and modifying the smaller coordinates does not change the maximum coordinate of any vector.
Thus, by Lemma~\ref{derand}, applying \autoref{fracalgVol} to this fractional schedule produces
an integral schedule of length at most
\begin{equation*}
    \Bar{M}+2+\frac{4\Bar{M}}{d}
      \leq\left(3+\frac{4}{d}\right)\Bar{M},
\end{equation*}
where we use $\Bar{M}\geq1$.

Moreover,
$OPT'\geq V\geq\widetilde V/(1+1/d)$ and
$OPT'\geq1$. Letting $M=\lceil2\widetilde V\rceil$ gives
\begin{equation*}
    M\leq\left(3+\frac{2}{d}\right)OPT'.
\end{equation*}
If $\Bar{M}=M$, then, for $d\geq2$,
\begin{equation*}
    \left(3+\frac{4}{d}\right)\Bar{M}
      \leq
    \left(3+\frac{4}{d}\right)
    \left(3+\frac{2}{d}\right)OPT'
      \leq 20 OPT'
      \leq 40 OPT.
\end{equation*}
As in Section~\ref{onlinesmall}, a standard folklore doubling argument for guessing $M$ increases this factor by at most four.
Therefore, the resulting deterministic online schedule has length at most $160OPT$.

\end{proof}

\section*{AI Disclosure} The authors used GPT-5.6 Sol (OpenAI) solely for English-language polishing, including improving grammar, clarity, and readability. The tool was not used to develop, verify, or generate any technical content, mathematical arguments, proofs, algorithms, results, or scientific claims in this paper. All technical contributions and their verification are entirely the work of the authors.

\newpage
\bibliography{main}

\end{document}